\documentclass{revtex4-2}
\usepackage[
	bookmarks  = false,
	colorlinks = true,
	linkcolor  = blue,
	urlcolor   = blue,
	citecolor  = blue]{hyperref}
\usepackage{amsmath,amssymb,amsthm,graphicx}
\usepackage{braket,bm,bbm}
\usepackage[T1]{fontenc}
\usepackage{mathrsfs}
\usepackage[normalem]{ulem}
\usepackage{bookmark}
\DeclareMathAlphabet{\pazocal}{OMS}{zplm}{m}{n}
\newtheorem{theorem}{Theorem}
\newtheorem{lemma}{Lemma}
\newtheorem{corollary}{Corollary}
\newtheorem{definition}{Definition}
\newtheorem{proposition}{Proposition}
\newtheorem{example}{Example}
\DeclareMathOperator{\Tr}{Tr}
\DeclareMathOperator{\supp}{supp}
\newcommand{\ketbra}[2]{\ket{#1}\!\bra{#2}}
\begin{document}
\title{Entanglement witness and nonlocality in confidence of measurement\\
from multipartite quantum state discrimination}
\author{Donghoon Ha}
\affiliation{Department of Applied Mathematics and Institute of Natural Sciences, Kyung Hee University, Yongin 17104, Republic of Korea}
\author{Jeong San Kim}
\email{freddie1@khu.ac.kr}
\affiliation{Department of Applied Mathematics and Institute of Natural Sciences, Kyung Hee University, Yongin 17104, Republic of Korea}
\begin{abstract}
We consider multipartite quantum state discrimination and provide a specific relation between the properties of entanglement witness and quantum nonlocality inherent in the confidence of measurement. 
We first provide the definition of the confidence of measurement as well as its useful properties for various types of multipartite measurement.
We show that globally-maximum confidence that cannot be achieved by local operations and classical communication strongly depends on the existence of a certain type of entanglement witness.
We also provide conditions for an upper bound on maximum of locally-achievable confidence. 
Finally, we establish a method in terms of entanglement witness to construct quantum state ensembles with nonlocal maximum confidence.
\end{abstract}
\maketitle
\section{Introduction}\label{sec:int}
Quantum nonlocality is an important feature of multipartite quantum systems without any classical counterpart\cite{horo2009,chid2013,brun2014}.
In discriminating multipartite quantum states,
nonlocal phenomenon occurs when a globally possible discrimination strategy cannot be realized only by \emph{local operations and classical communication}(LOCC)\cite{chit2014}.
The first nonlocality of quantum state discrimination was shown through orthogonal quantum states with local indistinguishability\cite{benn1999,ghos2001,walg2002}.
In general, orthogonal quantum states can be perfectly discriminated by using an appropriate measurement, whereas it is not true for nonorthogonal quantum states\cite{chef2000,berg2007,barn20091,bae2015}.
Nonlocality of quantum state discrimination can also occurs in discriminating nonorthogonal quantum states;
there exist some nonorthogonal quantum states where the globally-optimal discrimination cannot be realized using only LOCC measurements\cite{hels1969,pere1991,duan2007,chit2013}.

The phenomenon of nonlocality also arises in the correlation distributed in a multipartite quantum system.
Quantum entanglement is a nonlocal correlation that cannot be created only by LOCC\cite{horo2009}.
The nonlocal nature of entanglement is a useful resource in various quantum information tasks such as quantum cryptography and quantum teleportation\cite{eker1991,benn1993,chit2019}.
Thus, it is important and even necessary to detect the presence of entanglement inherent in multipartite quantum states.
\emph{Entanglement witness}(EW) is an entanglement-detecting observable providing a negative expectation value for some entangled states, whereas its expectation value is nonnegative for all separable states\cite{horo1996,terh2000,lewe2000,chru2014}.
Recently, it was shown that quantum nonlocality arising in multipartite quantum state discrimination is closely related to the existence of EW\cite{ha20231,ha20232};
in multipartite quantum state discrimination, the possible difference between globally-optimal and locally-optimal success probabilities can be directly verified by using the existence of EW.
These results establish possible relationship between different types of nonlocality from various quantum phenomena.

Here, we consider multipartite quantum state discrimination and provide a specific relation between the properties of EW and quantum nonlocality inherent in the confidence of measurement. 
We first provide the definition of the confidence of measurement as well as its useful properties for various types of multipartite measurement.
We show that globally-maximum confidence that cannot be achieved by LOCC meausurements strongly depends on the existence of a certain type of EW.
We also provide conditions for an upper bound on maximum of locally-achievable confidence. 
Finally, we establish a method in terms of EW to construct quantum state ensembles with nonlocal maximum confidence.

\section{Results}\label{sec:res}
\subsection{Confidence of measurement}
Let us consider the situation of discriminating the quantum states from the ensemble,
\begin{equation}\label{eq:ens}
\mathscr{E}=\{\eta_{i},\rho_{i}\}_{i=1}^{n},
\end{equation}
where the state $\rho_{i}$ is prepared with the \emph{nonzero} probability $\eta_{i}$ for each $i=1,\ldots,n$.
The state of a quantum system prepared from $\mathscr{E}$ is denoted by $\rho_{0}$, that is,
\begin{equation}\label{eq:ads}
\rho_{0}=\sum_{i=1}^{n}\eta_{i}\rho_{i}.
\end{equation}
We further consider the discrimination of the quantum state ensemble $\mathscr{E}$ in Eq.~\eqref{eq:ens} using a measurement $\mathscr{M}=\{M_{i}\}_{i=0}^{n}$. 
For each $i=1,\ldots,n$, we guess the prepared state to be $\rho_{i}$ if the measurement result is $M_{i}$.
On the other hand, the measurement result is inconclusive
when we obtain $M_{0}$.
The definitions and properties about separable operator and EW are provided in the ``\hyperref[sec:mtd]{Methods}'' Section. 

\begin{definition}
For a quantum state ensemble $\mathscr{E}=\{\eta_{i},\rho_{i}\}_{i=1}^{n}$ and a measurement $\mathscr{M}=\{M_{i}\}_{i=0}^{n}$, the \emph{confidence} of $\mathscr{M}$ to identify $\rho_{j}$ is the conditional probability 
\begin{equation}\label{eq:cfd}
\Pr(\rho_{j}|M_{j})=\frac{\eta_{j}\Tr(\rho_{j}M_{j})}{\Tr(\rho_{0}M_{j})},
\end{equation}
which is well defined only when
\begin{equation}\label{eq:wdc}
\Tr(\rho_{0} M_{j})\neq 0.
\end{equation}
For each $j=1,\ldots,n$, the \emph{maximum confidence} of measurements to identify $\rho_{j}$ is 
\begin{equation}\label{eq:cje}
\pazocal{C}_{j}(\mathscr{E})=\max_{\substack{\mathscr{M}\,\mathrm{with}\\
\Tr(\rho_{0} M_{j})\neq 0}}\Pr(\rho_{j}|M_{j}),
\end{equation}
where the maximum is taken over all possible measurements $\mathscr{M}=\{M_{i}\}_{i=0}^{n}$ with Eq.~\eqref{eq:wdc}.
\end{definition}

Here, we note that the confidence $\Pr(\rho_{j}|M_{j})$
in Eq.~\eqref{eq:cfd} is the probability that the prepared state is $\rho_{j}$ when the measurement result is $M_{j}$.
We further note that the maximum confidence $\pazocal{C}_{j}(\mathscr{E})$ in Eq.~\eqref{eq:cje} is already known as the largest eigenvalue of $\sqrt{\rho_{0}}^{-1}\eta_{j}\rho_{j}\sqrt{\rho_{0}}^{-1}$ where $\sqrt{E}$ is the positive square root of $E\in\mathbb{H}_{+}$, and $F^{-1}$ is the inverse of the Hermitian operator $F$ on its support\cite{crok2006,supp}.
From this fact, we can easily verify that
\begin{eqnarray}
\pazocal{C}_{j}(\mathscr{E})&=&\min\{q\in\mathbb{R}\,|\,q\mathbbm{1}-\sqrt{\rho_{0}}^{-1}\eta_{j}\rho_{j}\sqrt{\rho_{0}}^{-1}\in\mathbb{H}_{+}\}\nonumber\\
&=&\min\{q\in\mathbb{R}\,|\,q\rho_{0}-\eta_{j}\rho_{j}\in\mathbb{H}_{+}\},~j=1,\ldots,n,\label{eq:cjeq}
\end{eqnarray}
which can also be obtained by considering Eq.~\eqref{eq:cje} as a semidefinite program\cite{flat2022,lee2022,kenb2022}.

When the available measurements are restricted to separable measurements, we denote the maximum of achievable confidence by 
\begin{equation}\label{eq:sje}
\pazocal{S}_{j}(\mathscr{E})=\max_{\substack{\mathrm{Separable}\,\mathscr{M}\\
\mathrm{with}\, \Tr(\rho_{0} M_{j})\neq 0}}\Pr(\rho_{j}|M_{j}),~~j=1,\ldots,n.
\end{equation}
Similarly, we denote the maximum of confidence achievable using LOCC measurements as
\begin{equation}\label{eq:lje}
\pazocal{L}_{j}(\mathscr{E})=\max_{\substack{\mathrm{LOCC}\,\mathscr{M}\\
\mathrm{with}\, \Tr(\rho_{0} M_{j})\neq 0}}\Pr(\rho_{j}|M_{j}),~~j=1,\ldots,n.
\end{equation}
For each $j=1,\ldots,n$, it follows from the definitions of $\pazocal{C}_{j}(\mathscr{E})$, $\pazocal{S}_{j}(\mathscr{E})$ and $\pazocal{L}_{j}(\mathscr{E})$ that
\begin{equation}\label{eq:lsc}
0<\pazocal{L}_{j}(\mathscr{E})\leqslant 
\pazocal{S}_{j}(\mathscr{E})\leqslant
\pazocal{C}_{j}(\mathscr{E})\leqslant1.
\end{equation}
Here the first strict inequality is due to the positivity of $\Pr(\rho_{j}|M_{j})$ for a LOCC measurement $\{M_{i}\}_{i=0}^{n}$ where $M_{i}=\delta_{ij}\mathbbm{1}$ with the Kronecker delta $\delta_{ij}$.
The following theorem shows that the maximum of confidence achievable using separable measurements is also achievable by LOCC measurements.
\begin{theorem}\label{thm:lsm}
For a multipartite quantum state ensemble $\mathscr{E}=\{\eta_{i},\rho_{i}\}_{i=1}^{n}$
and each $j=1,\ldots,n$, we have
\begin{equation}\label{eq:sjmm}
\pazocal{L}_{j}(\mathscr{E})=
\pazocal{S}_{j}(\mathscr{E})=\max_{\substack{M\in\mathbb{SEP}\\ \Tr(\rho_{0} M)=1}}
\eta_{j}\Tr(\rho_{j}M),
\end{equation}
where the maximum is taken over all possible separable operator $M$ with $\Tr(\rho_{0} M)=1$ and $\rho_{0}$ is defined in Eq.~\eqref{eq:ads}.
\end{theorem}
\begin{proof}
For any separable measurement $\{M_{i}\}_{i=0}^{n}$ giving $\pazocal{S}_{j}(\mathscr{E})$, we have
\begin{equation}\label{eq:sjej}
\pazocal{L}_{j}(\mathscr{E})\leqslant
\pazocal{S}_{j}(\mathscr{E})=\frac{\eta_{j}\Tr(\rho_{j}M_{j})}{\Tr(\rho_{0} M_{j})}
=\eta_{j}\Tr(\rho_{j}M')
\leqslant\max_{\substack{M\in\mathbb{SEP}\\ \Tr(\rho_{0} M)=1}}
\eta_{j}\Tr(\rho_{j}M),
\end{equation}
where the first inequality is from Inequality~\eqref{eq:lsc},
the second equality is by letting $M'=M_{j}/\Tr(\rho_{0} M_{j})$,
and the last inequality holds due to 
$M'\in\mathbb{SEP}$ and $\Tr(\rho_{0} M')=1$.
Therefore, to prove Eq.~\eqref{eq:sjmm}, it is enough to show that
\begin{equation}\label{eq:sjed}
\max_{\substack{M\in\mathbb{SEP}\\ \Tr(\rho_{0} M)=1}}
\eta_{j}\Tr(\rho_{j}M)\leqslant \pazocal{L}_{j}(\mathscr{E}).
\end{equation}

Let $M^{\star}$ be an optimal operator realizing the left-hand side of Inequality~\eqref{eq:sjed}.
We first note that $\Tr M^{\star}$ is positive because
\begin{equation}\label{eq:tmej}
\Tr M^{\star}\geqslant \eta_{j}\Tr(\rho_{j} M^{\star})\geqslant\eta_{j}\Tr(\rho_{j}\mathbbm{1})=\eta_{j}>0,
\end{equation}
where the first inequality is from $\mathbbm{1}-\eta_{j}\rho_{j}\in\mathbb{H}_{+}$ and the second inequality holds due to the optimality of $M^{\star}$ among all separable operators $M$ satisfying $\Tr(\rho_{0}M)=1$, whereas $\mathbbm{1}$ is one of such operators.
From the positivity of $\Tr M^{\star}$, we can consider the measurement $\mathscr{M}=\{M_{i}\}_{i=0}^{n}$ consisting of
\begin{equation}\label{eq:smaw}
M_{i}=
\left\{
\begin{array}{ccc}
M^{\star}/\Tr M^{\star},&i=j,\\
\mathbbm{1}-M^{\star}/\Tr M^{\star},&i=0,\\
\mathbb{O},&\mathrm{otherwise.}
\end{array}
\right.
\end{equation}
Since $M^{\star}/\Tr M^{\star}$ is a separable state in $\mathbb{H}$, it follows from Lemma~\ref{lem:tos} that $\mathscr{M}$ is a LOCC measurement.
Therefore, we have
\begin{equation}\label{eq:sjin}
\max_{\substack{M\in\mathbb{SEP}\\ \Tr(\rho_{0} M)=1}}\eta_{j}\Tr(\rho_{j}M)
=\eta_{j}\Tr(\rho_{j}M^{\star})
=\frac{\eta_{j}\Tr(\rho_{j}M^{\star})}{\Tr(\rho_{0} M^{\star})}
=\frac{\eta_{j}\Tr(\rho_{j}M_{j})}{\Tr(\rho_{0} M_{j})}
\leqslant\pazocal{L}_{j}(\mathscr{E}),
\end{equation}
where the first and second equalities are from the definition of $M^{\star}$,
the last equality is by Eq.~\eqref{eq:smaw},
and the inequality holds from the definition of $\pazocal{L}_{j}(\mathscr{E})$ in Eq.~\eqref{eq:lje}.
\end{proof}

Now, let us consider the minimum quantities
\begin{equation}\label{eq:tsie}
\pazocal{Q}_{j}(\mathscr{E})=\min_{q\in\mathbb{R}_{j}(\mathscr{E})} q,~~j=1,\ldots,n,
\end{equation}
where
\begin{equation}\label{eq:rie}
\mathbb{R}_{j}(\mathscr{E})=\{q\in\mathbb{R}\,|\,q\rho_{0}-\eta_{j}\rho_{j}\in\mathbb{SEP}^{*}\}.
\end{equation}
Each $\pazocal{Q}_{j}(\mathscr{E})$ in Eq.~\eqref{eq:tsie} is
an upper bound of $\pazocal{S}_{j}(\mathscr{E})$ in Eq.~\eqref{eq:lje} because
\begin{equation}\label{eq:stsi}
\pazocal{S}_{j}(\mathscr{E})=\max_{\substack{M\in\mathbb{SEP}\\ \Tr(\rho_{0} M)=1}}\eta_{j}\Tr(\rho_{j}M)
\leqslant\max_{\substack{M\in\mathbb{SEP}\\ \Tr(\rho_{0} M)=1}}\Big(
\eta_{j}\Tr(\rho_{j}M)+\Tr[(\pazocal{Q}_{j}(\mathscr{E})\rho_{0}-\eta_{j}\rho_{j})M]\Big)
=\pazocal{Q}_{j}(\mathscr{E}),
\end{equation}
where the first equality is from Theorem~\ref{thm:lsm} and the inequality is due to $\pazocal{Q}_{j}(\mathscr{E})\rho_{0}-\eta_{j}\rho_{j}\in\mathbb{SEP}^{*}$.

For an ensemble $\mathscr{E}=\{\eta_{i},\rho_{i}\}_{i=1}^{n}$ and each $j=1,\ldots,n$, the following theorem shows that $\pazocal{S}_{j}(\mathscr{E})$ is equal to $\pazocal{Q}_{j}(\mathscr{E})$.
The proof of Theorem~\ref{thm:sts} is provided in the ``\hyperref[sec:mtd]{Methods}'' Section.
\begin{theorem}\label{thm:sts}
For a multipartite quantum state ensemble $\mathscr{E}=\{\eta_{i},\rho_{i}\}_{i=1}^{n}$ and each $j=1,\ldots,n$, we have
\begin{equation}\label{eq:stsj}
\pazocal{S}_{j}(\mathscr{E})=\pazocal{Q}_{j}(\mathscr{E}).
\end{equation}
\end{theorem}

For a multipartite quantum state ensemble $\mathscr{E}=\{\eta_{i},\rho_{i}\}_{i=1}^{n}$ and each $j=1,\ldots,n$, the following theorem provides a necessary and sufficient condition for $q\in\mathbb{R}_{j}(\mathscr{E})$ to give $\pazocal{Q}_{j}(\mathscr{E})$.
\begin{theorem}\label{thm:wps}
For a multipartite quantum state ensemble $\mathscr{E}=\{\eta_{i},\rho_{i}\}_{i=1}^{n}$ and $q\in\mathbb{R}_{j}(\mathscr{E})$ with $j\in\{1,\ldots,n\}$, 
\begin{equation}\label{eq:ljq}
q=\pazocal{Q}_{j}(\mathscr{E})
\end{equation}
if and only if there exists a separable state $\sigma$ satisfying
\begin{equation}\label{eq:omc}
\Tr[\sigma(q\rho_{0}-\eta_{j}\rho_{j})]=0,~
\Tr(\sigma\rho_{0})>0.
\end{equation}
\end{theorem}
\begin{proof}
Let us first suppose that Eq.~\eqref{eq:ljq} is satisfied.
From Theorems~\ref{thm:lsm} and \ref{thm:sts} together with Eq.~\eqref{eq:ljq},  we have
\begin{equation}\label{eq:sjq}
\max_{\substack{M\in\mathbb{SEP}\\ \Tr(\rho_{0} M)=1}}
\eta_{j}\Tr(\rho_{j}M)=q.
\end{equation}
Equation~\eqref{eq:sjq} implies that there exists $M^{\star}\in\mathbb{SEP}$ with
\begin{equation}\label{eq:msc}
\Tr(\rho_{0}M^{\star})=1,~\eta_{j}\Tr(\rho_{j}M^{\star})=q.
\end{equation}
Thus, the separable state $\sigma=M^{\star}/\Tr M^{\star}$ satisfies
\begin{equation}\label{eq:sqe}
\Tr[\sigma(q\rho_{0}-\eta_{j}\rho_{j})]=\frac{q-q}{\Tr M^{\star}}=0,~~
\Tr(\sigma\rho_{0})=\frac{1}{\Tr M^{\star}}>0,
\end{equation}
where the inequality follows from $M^{\star}\in\mathbb{SEP}$ and $\Tr(\rho_{0}M^{\star})=1$.

Conversely, let us assume there exists a separable state $\sigma$ satisfying Eq.~\eqref{eq:omc}.
By letting $M=\sigma/\Tr(\sigma\rho_{0})\in\mathbb{SEP}$, we have
\begin{equation}\label{eq:qme}
q=\eta_{j}\Tr(\rho_{j}M)\leqslant\pazocal{S}_{j}(\mathscr{E})=\pazocal{Q}_{j}(\mathscr{E})\leqslant q,
\end{equation}
where the first equality is due to the equality in Eq.~\eqref{eq:omc}, 
the first inequality is by Theorem~\ref{thm:lsm}, 
the second equality is from Theorem~\ref{thm:sts}, and the second inequality follows from the definition of $\pazocal{Q}_{j}(\mathscr{E})$ in Eq.~\eqref{eq:tsie} together with $q\in\mathbb{R}_{j}(\mathscr{E})$.
Thus, Eq.~\eqref{eq:ljq} is satisfied.
\end{proof}

\subsection{Nonlocal maximum confidence}
For a multipartite quantum state ensemble $\mathscr{E}=\{\eta_{i},\rho_{i}\}_{i=1}^{n}$ and a measurement $\mathscr{M}=\{M_{i}\}_{i=0}^{n}$, we say that the confidence of $\mathscr{M}$ to identify $\rho_{j}$ is \emph{nonlocal} if it cannot be achieved by LOCC measurements, that is,
\begin{equation}\label{eq:nlc}
\pazocal{L}_{j}(\mathscr{E})<\Pr(\rho_{j}|M_{j}),
\end{equation}
where $\Pr(\rho_{j}|M_{j})$ and $\pazocal{L}_{j}(\mathscr{E})$ are defined in Eqs.~\eqref{eq:cfd} and \eqref{eq:lje}, respectively.
From this perspective, the maximum confidence of measurements to identify $\rho_{j}$ in Eq.~\eqref{eq:cje} is called \emph{nonlocal} if
\begin{equation}\label{eq:almc}
\pazocal{L}_{j}(\mathscr{E})<\pazocal{C}_{j}(\mathscr{E}).
\end{equation}
From Theorems~\ref{thm:lsm} and \ref{thm:sts}, we note that Inequality~\eqref{eq:almc} is equivalent to 
\begin{equation}\label{eq:smc}
\pazocal{Q}_{j}(\mathscr{E})<\pazocal{C}_{j}(\mathscr{E}).
\end{equation}
The following theorem provides a necessary and sufficient condition for Inequality~\eqref{eq:smc} in terms of EW.
\begin{theorem}\label{thm:lcj}
For a multipartite quantum state ensemble $\mathscr{E}=\{\eta_{i},\rho_{i}\}_{i=1}^{n}$, 
$q\in\mathbb{R}$ and each $j=1,\ldots,n$, 
\begin{equation}\label{eq:llcj}
\pazocal{Q}_{j}(\mathscr{E})\leqslant q<\pazocal{C}_{j}(\mathscr{E})
\end{equation}
if and only if $q\rho_{0}-\eta_{j}\rho_{j}$ is an EW.
\end{theorem}
\begin{proof}
Let us suppose that Inequality~\eqref{eq:llcj} holds.
From the definition of $\pazocal{Q}_{j}(\mathscr{E})$, we have
\begin{equation}
q\rho_{0}-\eta_{j}\rho_{j}\in\mathbb{SEP}^{*}.
\end{equation}
Furthermore, Eq.~\eqref{eq:cjeq} and Inequality~\eqref{eq:llcj} lead us to
\begin{equation}
q\rho_{0}-\eta_{j}\rho_{j}\notin\mathbb{H}_{+}.
\end{equation}
Thus, $q\rho_{0}-\eta_{j}\rho_{j}$ is an EW.

Conversely, let us assume that $q\rho_{0}-\eta_{j}\rho_{j}$ is an EW.
Due to the definition of $\pazocal{Q}_{j}(\mathscr{E})$, we have
\begin{equation}\label{eq:ssjs}
\pazocal{Q}_{j}(\mathscr{E})\leqslant q.
\end{equation}
Furthermore, it follows from $q\rho_{0}-\eta_{j}\rho_{j}\notin\mathbb{H}_{+}$ and Eq.~\eqref{eq:cjeq} that
\begin{equation}\label{eq:cqci}
q<\pazocal{C}_{j}(\mathscr{E}).
\end{equation}
Thus, Inequalities~\eqref{eq:ssjs} and \eqref{eq:cqci} lead us to Eq.~\eqref{eq:llcj}.
\end{proof}

For an ensemble $\mathscr{E}=\{\eta_{i},\rho_{i}\}_{i=1}^{n}$ with $\pazocal{Q}_{j}(\mathscr{E})<\pazocal{C}_{j}(\mathscr{E})$ for some $j\in\{1,\ldots,n\}$, 
it follows from Theorems~\ref{thm:wps} and \ref{thm:lcj} that a real number $q$ is equal to $\pazocal{Q}_{j}(\mathscr{E})$ only if $q\rho_{0}-\eta_{j}\rho_{j}$ is an EW satisfying the first condition in Eq.~\eqref{eq:omc} for some separable state $\sigma$, that is, a weakly-optimal EW.
Moreover, the converse is true when $\rho_{0}$ is full rank because any separable state $\sigma$ satisfies the second condition in Eq.~\eqref{eq:omc}.
However, weak optimality of $q\rho_{0}-\eta_{j}\rho_{j}$ does not generally mean that $q=\pazocal{Q}_{j}(\mathscr{E})$.
In other words, there exists an ensemble $\mathscr{E}=\{\eta_{i},\rho_{i}\}_{i=1}^{n}$ where
$q\rho_{0}-\eta_{j}\rho_{j}$ is a weakly-optimal EW for some $q\in\mathbb{R}$ and $j\in\{1,\ldots,n\}$ with
\begin{equation}\label{eq:qlle}
q>\pazocal{Q}_{j}(\mathscr{E})
\end{equation}
which is illustrated in the following example.
\begin{example}
Let us consider the qubit-qutrit state ensemble $\mathscr{E}=\{\eta_{i},\rho_{i}\}_{i=1}^{3}$ consisting of three states,
\begin{align}
\eta_{1}=\tfrac{1}{3},&~~\rho_{1}=\tfrac{1}{2}\ketbra{00}{00}
+\tfrac{1}{2}\ketbra{12}{12},\nonumber\\
\eta_{2}=\tfrac{1}{3},&~~\rho_{2}=\tfrac{1}{2}\ketbra{02}{02}
+\tfrac{1}{2}\ketbra{10}{10},\label{eq:cex}\\
\eta_{3}=\tfrac{1}{3},&~~\rho_{3}=\ketbra{\xi}{\xi},~\ket{\xi}=\tfrac{1}{\sqrt{2}}\ket{00}+\tfrac{1}{\sqrt{2}}\ket{12}.\nonumber
\end{align}
\end{example}

For $q\in\mathbb{R}$, it is straightforward to verify that
\begin{equation}\label{eq:qrmo}
q\rho_{0}-\eta_{1}\rho_{1}
=\tfrac{1}{6}\big[
(2q-1)\ketbra{00}{00}+(2q-1)\ketbra{12}{12}
+q\ketbra{00}{12}+q\ketbra{12}{00}
+q\ketbra{02}{02}+q\ketbra{10}{10}\big].
\end{equation}
From Eq.~\eqref{eq:qrmo}, we can easily see that $q\rho_{0}-\eta_{1}\rho_{1}\in\mathbb{H}_{+}$
if and only if $q\geqslant1$. 
Therefore, Eq.~\eqref{eq:cjeq} leads us to
\begin{equation}\label{eq:ecoo}
\pazocal{C}_{1}(\mathscr{E})=1.
\end{equation}

To obtain $\pazocal{Q}_{1}(\mathscr{E})$, we use the fact that a Hermitian operator is block positive if its partial transposition is positive semidefinite\cite{pere1996,pptp}.
If $q\geqslant\frac{1}{2}$, then $q\rho_{0}-\eta_{1}\rho_{1}$ in Eq.~\eqref{eq:qrmo} is block positive because its partial transposition is positive semidefinite.
On the other hand, Definition~\ref{def:bpos} implies that $q\rho_{0}-\eta_{1}\rho_{1}$ is not block positive for $q<\frac{1}{2}$ because
\begin{equation}\label{eq:too}
\bra{00}(q\rho_{0}-\eta_{1}\rho_{1})\ket{00}=2q-1<0.
\end{equation}
Thus, the definition of $\pazocal{Q}_{1}(\mathscr{E})$ in Eq.~\eqref{eq:tsie} leads us to
\begin{equation}\label{eq:leco}
\pazocal{Q}_{1}(\mathscr{E})=\tfrac{1}{2}.
\end{equation}
Moreover, $q\rho_{0}-\eta_{1}\rho_{1}$ is a weakly-optimal EW for all $q\in[\frac{1}{2},1)$
because it is in $\mathbb{SEP}^{*}\setminus\mathbb{H}_{+}$ and 
\begin{equation}\label{eq:etwo}
\bra{01}(q\rho_{0}-\eta_{1}\rho_{1})\ket{01}=
\bra{11}(q\rho_{0}-\eta_{1}\rho_{1})\ket{11}=0.
\end{equation}
For $\frac{1}{2}<q<1$, $q\rho_{0}-\eta_{1}\rho_{1}$ is a weakly-optimal EW, but $q\neq\pazocal{Q}_{1}(\mathscr{E})$.

For a two-qubit state ensemble $\mathscr{E}=\{\eta_{i},\rho_{i}\}_{i=1}^{n}$ and each $j=1,\ldots,n$, the following corollary show that a real number $q$ is equal to $\pazocal{Q}_{j}(\mathscr{E})$ if $q\rho_{0}-\eta_{j}\rho_{j}$ is a weakly-optimal EW. 
\begin{corollary}\label{cor:lwo}
For a two-qubit state ensemble $\mathscr{E}=\{\eta_{i},\rho_{i}\}_{i=1}^{n}$, $q\in\mathbb{R}$ and each $j\in\{1,\ldots,n\}$, 
\begin{equation}\label{eq:lpq}
\pazocal{Q}_{j}(\mathscr{E})=q<\pazocal{C}_{j}(\mathscr{E})
\end{equation}
if and only if $q\rho_{0}-\eta_{j}\rho_{j}$ is a weakly-optimal EW.
Moreover, $\pazocal{C}_{j}(\mathscr{E})$ is achievable locally
when $\rho_{0}$ is not full rank.
\end{corollary}
\begin{proof}
The necessity of the first statement has already been mentioned.
To prove the sufficiency of the first statement, let us assume that $q\rho_{0}-\eta_{j}\rho_{j}$ is a weakly-optimal EW. This assumption implies that Eq.~\eqref{eq:llcj} is satisfied and there exists a separable state $\sigma$ satisfying the first condition in Eq.~\eqref{eq:omc}. The second condition in Eq.~\eqref{eq:omc} is also satisfied otherwise the EW $q\rho_{0}-\eta_{j}\rho_{j}$ has zero eigenvalue\cite{augu2008,sarb2008}. 
Thus, Theorems~\ref{thm:wps} and \ref{thm:lcj} lead us to Eq.~\eqref{eq:lpq}.

Let us consider the case where $\rho_{0}$ is not full rank. In this case, $q\rho_{0}-\eta_{j}\rho_{j}$ is also not full rank due to the definition of $\rho_{0}$ in Eq.~\eqref{eq:ads}, and it is not an EW because every two-qubit EW do not have zero eigenvalue\cite{augu2008,sarb2008}.
Thus, the second statement is true.
\end{proof}

Now, we establish a systematical way in terms of EW to construct quantum state ensembles where some maximum confidence is nonlocal.
For an EW $W$, let us consider the quantum state ensemble $\mathscr{E}=\{\eta_{i},\rho_{i}\}_{i=1}^{2}$ consisting of two orthogonal states,
\begin{equation}\label{eq:ftsx}
\eta_{1}=\frac{\Tr W_{-}}{\Tr W_{+}+\Tr W_{-}},~~\rho_{1}=\frac{W_{-}}{\Tr W_{-}},~~
\eta_{2}=\frac{\Tr W_{+}}{\Tr W_{+}+\Tr W_{-}},~~\rho_{2}=\frac{W_{+}}{\Tr W_{+}}
\end{equation}
where $W_{+}$ and $W_{-}$ are the orthogonal positive-semidefinite operators satisfying
\begin{equation}\label{eq:abps}
W=W_{+}-W_{-}.
\end{equation}

For $q\in\mathbb{R}$, we have
\begin{equation}\label{eq:ewce}
q\rho_{0}-\eta_{1}\rho_{1}=\tfrac{qW_{+}+(q-1) W_{-}}{\Tr W_{+}+\Tr W_{-}}.
\end{equation}
Since $q\rho_{0}-\eta_{1}\rho_{1}$ is proportional to the EW $W$ for $q=\frac{1}{2}$, it follows from Theorem~\ref{thm:lcj} that
\begin{equation}\label{eq:loco}
\pazocal{Q}_{1}(\mathscr{E})\leqslant\tfrac{1}{2}<\pazocal{C}_{1}(\mathscr{E}).
\end{equation}
Thus, the maximum confidence of measurements to identify $\rho_{1}$ is nonlocal.
In this case, we can also easily see that $q=1$ is the minimum value satisfying $q\rho_{0}-\eta_{1}\rho_{1}\in\mathbb{H}_{+}$, so Eq.~\eqref{eq:cjeq} leads us to
\begin{equation}\label{eq:loc1}
\pazocal{C}_{1}(\mathscr{E})=1.
\end{equation}

Now, we provide another way to construct quantum state ensembles where every maximum confidence is nonlocal.
For a given set of EWs $\{W_{i}\}_{i=1}^{n}$, let us consider
\begin{equation}\label{eq:wdse}
W=\sum_{i=1}^{n}W_{i},
\end{equation}
and denote the smallest eigenvalue of $W$ as $\epsilon$. For the case when $\epsilon>0$, we can define
a quantum state ensemble $\mathscr{E}=\{\eta_{i},\rho_{i}\}_{i=1}^{n}$ with
\begin{equation}\label{eq:ewey}
\eta_{i}=\frac{\Tr(\lambda_{i}W-\epsilon W_{i})}{\Tr(\lambda W-\epsilon W)},~
\rho_{i}=\frac{\lambda_{i}W-\epsilon W_{i}}{\Tr(\lambda_{i}W-\epsilon W_{i})}
\end{equation}
where $\lambda_{i}$ is the largest eigenvalue of $W_{i}$ for each $i=1,\ldots,n$ and 
\begin{equation}\label{eq:lsli}
\lambda:=\sum_{i=1}^{n}\lambda_{i}.
\end{equation}

For $q\in\mathbb{R}$ and each $j=1,\ldots,n$, a straightforward calculation leads us to
\begin{equation}\label{eq:qrms}
q\rho_{0}-\eta_{j}\rho_{j}=\tfrac{[q(\lambda-\epsilon)-\lambda_{j}]W+\epsilon W_{j}}{\Tr(\lambda W+\epsilon W)}.
\end{equation}
Since $q\rho_{0}-\eta_{j}\rho_{j}$ is proportional to the EW $W_{j}$ for $q=\frac{\lambda_{j}}{\lambda-\epsilon}$, it follows from Eq.~\eqref{eq:qrms} together with Theorem~\ref{thm:lcj} that
\begin{equation}\label{eq:tfee}
\pazocal{Q}_{j}(\mathscr{E})\leqslant\tfrac{\lambda_{j}}{\lambda-\epsilon}<\pazocal{C}_{j}(\mathscr{E}).
\end{equation}
Thus, the maximum confidence of measurements to identify $\rho_{j}$ is nonlocal for every $j=1,\ldots,n$.

$W$ in Eq.~\eqref{eq:wdse} is a full-rank positive-semidefinite operator because its smallest eigenvalue $\epsilon$ is positive.
For each $j=1,\ldots,n$, let $\delta_{j}$ be the absolute value of the smallest negative eigenvalue of $\sqrt{W}^{-1}W_{j}\sqrt{W}^{-1}$.
In this case, we can also easily see that $q=\frac{\lambda_{j}+\epsilon\delta_{j}}{\lambda-\epsilon}$ is the minimum value satisfying that $q\rho_{0}-\eta_{j}\rho_{j}$ in Eq.~\eqref{eq:qrms} is positive semidefinite.
Thus, Eq.~\eqref{eq:cjeq} leads us to
\begin{equation}\label{eq:tfea}
\pazocal{C}_{j}(\mathscr{E})=\tfrac{\lambda_{j}+\epsilon\delta_{j}}{\lambda-\epsilon}.
\end{equation}

We also note that $\rho_{0}$ is proportional to $W$ in Eq.~\eqref{eq:wdse} which is full rank, therefore the second condition of Eq.~\eqref{eq:omc} in Theorem~\ref{thm:wps} is satisfied for every separable state $\sigma$.
Thus, for the case when $W_{j}$ is a weakly-optimal EW, it follows from Theorem~\ref{thm:wps} that
\begin{equation}\label{eq:ljeq}
\pazocal{Q}_{j}(\mathscr{E})=\tfrac{\lambda_{j}}{\lambda-\epsilon}.
\end{equation}

\section{Discussion}\label{sec:dis}
We have considered the discrimination of multipartite quantum states and established a specific relation between the properties of EW and quantum nonlocality inherent in the confidence of measurement. 
We have first provided the definition of the confidence of measurement as well as its useful properties for various types of multipartite measurement(Theorems~\ref{thm:lsm} and \ref{thm:sts}).
We have shown that nonlocal maximum confidence strongly depends on the existence of a certain type of EW(Theorem~\ref{thm:lcj} and Corollary~\ref{cor:lwo}).
We have further established conditions for an upper bound on maximum of locally-achievable confidence(Theorem~\ref{thm:wps}). 
Finally, we have provided a method in terms of EW to construct quantum state ensembles with nonlocal maximum confidence.

EW is also known to play an important role in multipartite quantum state discrimination, namely, minimum-error discrimination by separable measurements.
The minimum-error discrimination of an ensemble $\mathscr{E}$ in Eq.~\eqref{eq:ens} is to maximize the average probability of correctly guessing the prepared state from $\mathscr{E}$, that is,
\begin{equation}\label{eq:gupro}
p_{\mathrm{G}}(\mathscr{E})=\max_{\mathscr{M}}\sum_{i=1}^{n}\eta_{i}\Tr(\rho_{i}M_{i}),
\end{equation}
where the maximum is taken over all possible measurements $\mathscr{M}=\{M_{i}\}_{i=0}^{n}$\cite{hels1969}. When the available measurements are limited to separable measurements, the maximum success probability is denoted by
\begin{equation}\label{eq:psep}
p_{\mathrm{SEP}}(\mathscr{E})=\max_{\mathrm{Separable}\,\mathscr{M}}\sum_{i=1}^{n}\eta_{i}\Tr(\rho_{i}M_{i}).
\end{equation}
The minimum-error discrimination of $\mathscr{E}$ is realized by separable measurements if
\begin{equation}\label{eq:mebsp}
p_{\mathrm{SEP}}(\mathscr{E})=p_{\mathrm{G}}(\mathscr{E}).
\end{equation}

It was recently shown that the existence of specific type of EW guarantees the minimum-error discrimination of a quantum state ensemble by separable measurements\cite{ha20231}.
\begin{proposition}\label{prop:pspg}
For a multipartite quantum state ensemble $\mathscr{E}=\{\eta_{i},\rho_{i}\}_{i=1}^{n}$ and each $j=1,\ldots,n$, 
\begin{equation}\label{eq:pspg}
p_{\mathrm{SEP}}(\mathscr{E})=\eta_{j}<p_{\mathrm{G}}(\mathscr{E})
\end{equation}
if and only if $\eta_{j}\rho_{j}-\eta_{i}\rho_{i}$ is block positive for all $i=1,\ldots,n$
and there exists an EW in $\{\eta_{j}\rho_{j}-\eta_{i}\rho_{i}\}_{i=1}^{n}$.
\end{proposition}

Here, we further provide the mutual implication between nonlocal maximum confidence and minimum-error discrimination by separable measurements in discriminating multipartite quantum states.
\begin{theorem}\label{thm:lumm}
For a multipartite quantum state ensemble $\mathscr{E}=\{\eta_{i},\rho_{i}\}_{i=1}^{n}$ and each $j=1,\ldots,n$, the maximum confidence of measurements to identify $\rho_{j}$ in $\mathscr{E}$ is nonlocal, that is, 
\begin{equation}\label{eq:dqcj}
\pazocal{Q}_{j}(\mathscr{E})<\pazocal{C}_{j}(\mathscr{E})
\end{equation}
if and only if there exists a real number $r\in(0,1)$ satisfying
\begin{equation}\label{eq:psgr}
p_{\mathrm{SEP}}(\mathscr{E}_{j}^{r})=r<p_{\mathrm{G}}(\mathscr{E}_{j}^{r})
\end{equation}
where 
\begin{equation}\label{eq:ejr}
\mathscr{E}_{j}^{r}=\{r,\rho_{0};1-r,\rho_{j}\}
\end{equation}
is the two-state ensemble consisting of the states
$\rho_{0}$ and $\rho_{j}$ with the corresponding probabilities $r$ and $1-r$, respectively.
\end{theorem}
\begin{proof}
Let us suppose Inequality~\eqref{eq:dqcj} holds and $q$ is a real number satisfying $\pazocal{Q}_{j}(\mathscr{E})\leqslant q<\pazocal{C}_{j}(\mathscr{E})$.
Due to Theorem~\ref{thm:lcj}, $q\rho_{0}-\eta_{j}\rho_{j}$ is an EW, therefore it follows from Proposition~\ref{prop:pspg} that Eq.~\eqref{eq:psgr} holds for $r=\frac{q}{q+\eta_{j}}$.

Conversely, assume that Eq.~\eqref{eq:psgr} is satisfied for some $r\in(0,1)$. 
From Proposition~\ref{prop:pspg}, we can see that $q\rho_{0}-\eta_{j}\rho_{j}$ is an EW for $q=r\eta_{j}/(1-r)$. 
Thus, Theorem~\ref{thm:lcj} leads us to Inequality~\eqref{eq:dqcj}.
\end{proof}

In the minimum-error discrimination of a multipartite quantum state ensemble $\mathscr{E}=\{\eta_{i},\rho_{i}\}_{i=1}^{n}$, quantum nonlocality occurs when the optimal success probability $p_{\mathrm{G}}(\mathscr{E})$ in Eq.~\eqref{eq:gupro} cannot be achieved by LOCC measurements, that is,
\begin{equation}\label{eq:plpg}
p_{\mathrm{L}}(\mathscr{E})< p_{\mathrm{G}}(\mathscr{E})
\end{equation}
where $p_{\mathrm{L}}(\mathscr{E})$ is the maximum of success probability that can be obtained using LOCC measurements, that is,
\begin{equation}\label{eq:ple}
p_{\mathrm{L}}(\mathscr{E})=\max_{\mathrm{LOCC}\,\mathscr{M}}\sum_{i=1}^{n}\eta_{i}\Tr(\rho_{i}M_{i}).
\end{equation}
Since the definitions of $p_{\mathrm{SEP}}(\mathscr{E})$ and $p_{\mathrm{L}}(\mathscr{E})$ in Eqs.~\eqref{eq:psep} and \eqref{eq:ple} imply $p_{\mathrm{L}}(\mathscr{E})\leqslant p_{\mathrm{SEP}}(\mathscr{E})$, Inequality~\eqref{eq:plpg} holds if $p_{\mathrm{SEP}}(\mathscr{E})< p_{\mathrm{G}}(\mathscr{E})$.
Due to this fact, Eq.~\eqref{eq:psgr} of Theorem~\ref{thm:lumm} guarantees $p_{\mathrm{L}}(\mathscr{E}_{j}^{r})<p_{\mathrm{G}}(\mathscr{E}_{j}^{r})$.
In other words, Theorem~\ref{thm:lumm} tells us that the nonlocality in the confidence of measurement characterized by Inequality~\eqref{eq:dqcj} always guarantees the nonlocality arising in the minimum-error discrimination of quantum states in Eq.~\eqref{eq:psgr}.
Thus, we have established a directly connection between distinct types of nonlocality arising from multipartite quantum state discrimination.

As our results provide a specific relation between EW and nonlocal maximum confidence, it is natural to investigate the relationship between EW and quantum nonlocality arising in \emph{optimal} maximum-confidence discrimination which is to maximize the success probability of discrimination over all possible maximum-confidence measurements.
It is also a good future work to investigate the possible implications and applications of nonlocality in the confidence of measurement such as data hiding or secret sharing to construct secure information network.

\section{Methods}\label{sec:mtd}
For a multipartite Hilbert spaces $\mathscr{H}=\bigotimes_{k=1}^{m}\mathbb{C}^{d_{k}}$ with positive integers $m\geqslant2$ and $d_{1},\ldots,d_{m}$, let us denote by $\mathbb{H}$ the set of all Hermitian operators acting on $\mathscr{H}$.
We also denote the set of all positive-semidefinite operators in $\mathbb{H}$ by
\begin{equation}\label{eq:hplu}
\mathbb{H}_{+}=\{E\in\mathbb{H}\,|\,\bra{v}E\ket{v}\geqslant0~~\forall\ket{v}\in\mathscr{H}\}.
\end{equation}

A multipartite quantum state is described by $\rho\in\mathbb{H}_{+}$ with $\Tr\rho=1$ and a measurement is expressed by $\{M_{i}\}_{i}\subseteq\mathbb{H}_{+}$ satisfying $\sum_{i}M_{i}=\mathbbm{1}$ where $\mathbbm{1}$ is the identity operator in $\mathbb{H}$.
When $\rho$ is measured in $\{M_{i}\}_{i}$,
$M_{i}$ is detected with the probability $\Tr(\rho M_{i})$.

\begin{definition}\label{def:sep}
$E\in\mathbb{H}_{+}$ is called \emph{separable} if it can be described by a conic combination of product states, that is,
\begin{equation}\label{eq:sepo}
E=\sum_{l}p_{l}\bigotimes_{k=1}^{m}\sigma_{l}^{(k)}
\end{equation}
where $\{p_{l}\}_{l}$ is a set of nonnegative real numbers and
$\{\sigma_{l}^{(k)}\}_{l}$ is a set of states acting on $\mathbb{C}^{d_{k}}$ for each $k=1,\ldots,m$.
\end{definition}
\noindent We denote the set of all separable operators in $\mathbb{H}_{+}$ by
\begin{equation}\label{eq:sepd}
\mathbb{SEP}=\{E\in\mathbb{H}_{+}\,|\,E:\mathrm{separable}\}.
\end{equation}

A measurement $\{M_{i}\}_{i}$ is called a \emph{separable measurement} if 
$M_{i}\in\mathbb{SEP}$ for all $i$. 
We also say that a measurement is a \emph{LOCC measurement} if it can be realized by LOCC.
Note that every LOCC measurement is a separable measurement\cite{chit2014}.

The following lemma provide a way to obtain a LOCC measurement from a separable state.
\begin{lemma}\label{lem:tos}
For $E\in\mathbb{SEP}$ with $\Tr E=1$, the measurement $\{E,\mathbbm{1}-E\}$ is a LOCC measurement.
\end{lemma}
\begin{proof}
Since $E\in\mathbb{SEP}$, $E$ can be described by the form in Eq.~\eqref{eq:sepo}.
In this case, $\{p_{l}\}_{l}$ becomes a probability distribution due to $\Tr E=1$.
Let us consider the following LOCC process;
with the probability $p_{l}$ for each $l$, the $1$st,$\ldots$,$m$th parties perform the local measurements $\{\sigma_{l}^{(1)},\mathbbm{1}_{1}-\sigma_{l}^{(1)}\},\ldots,\{\sigma_{l}^{(m)},\mathbbm{1}_{m}-\sigma_{l}^{(m)}\}$, respectively, where $\mathbbm{1}_{k}$ is the identity operator acting on $\mathbb{C}^{d_{k}}$ for each $k=1,\ldots,m$.
The parties then use classical communication to answer ``0'' if the measurement result is $\bigotimes_{k=1}^{m}\sigma_{l}^{(k)}$, and ``1'' otherwise.
This process with two outcomes ``0'' and ``1'' is equivalent to the measurement $\{M_{0}=E,M_{1}=\mathbbm{1}-E\}$.
Therefore, the measurement $\{E,\mathbbm{1}-E\}$ is a LOCC measurement.
\end{proof}

\begin{definition}\label{def:bpos}
$E\in\mathbb{H}$ is called \emph{block positive} if
\begin{equation}\label{eq:bpd}
\Tr(EF)\geqslant0
\end{equation}
for all $F\in\mathbb{SEP}$.
\end{definition}
\noindent We denote by $\mathbb{SEP}^{*}$ the set of all block-positive operators in $\mathbb{H}$, that is,
\begin{equation}\label{eq:seps}
\mathbb{SEP}^{*}=\{E\in\mathbb{H}\,|\,E:\mathrm{block~positive}\}.
\end{equation}
$\mathbb{SEP}^{*}$ is the dual cone of $\mathbb{SEP}$, but $\mathbb{SEP}$ is also the dual cone of $\mathbb{SEP}^{*}$ because $\mathbb{SEP}$ is convex and closed\cite{boyd2004}, that is,
\begin{equation}\label{eq:sds}
\{E\in\mathbb{H}\,|\,\Tr(EF)\geqslant0~\forall F\in\mathbb{SEP}^{*}\}=\mathbb{SEP}.
\end{equation}

We note that
\begin{equation}\label{eq:subs}
\mathbb{SEP}\subseteq\mathbb{H}_{+}\subseteq\mathbb{SEP}^{*}\subseteq\mathbb{H}
\end{equation}
where the first and last inclusions are from the definitions of $\mathbb{SEP}$ and $\mathbb{SEP}^{*}$, respectively, and the second inclusion is from the fact that $\Tr(EF)\geqslant0$ for all $E\in\mathbb{H}_{+}$ and all $F\in\mathbb{SEP}$. We also note that
\begin{equation}\label{eq:twz}
\Tr E>0
\end{equation}
for any $E\in\mathbb{SEP}^{*}$ with $E\neq\mathbb{O}$ where $\mathbb{O}$ is the zero operator in $\mathbb{H}$\cite{ha20222}.

\begin{definition}\label{def:enwt}
$W\in\mathbb{H}$ is called an \emph{EW} if $\Tr(WE)\geqslant0$ for all $E\in\mathbb{SEP}$ but $\Tr(WF)<0$ for some $F\in\mathbb{H}_{+}\setminus\mathbb{SEP}$, or equivalently,
\begin{equation}\label{eq:ewd}
W\in\mathbb{SEP}^{*}\setminus\mathbb{H}_{+}.
\end{equation}
\end{definition}

An EW $W$ is called \emph{optimal} if there is no other EW detecting more entangled states than $W$ does; in other words, there does not exist $W'\in\mathbb{SEP}^{*}\setminus\mathbb{H}_{+}$ satisfying
$\Tr(W'E)<0$ for all $E\in\mathbb{H}_{+}\setminus\mathbb{SEP}$ with $\Tr(WE)<0$ and $\Tr(W'F)<0$ for some $F\in\mathbb{H}_{+}\setminus\mathbb{SEP}$ with $\Tr(WF)\geqslant 0$\cite{lewe2000}.
An EW $W$ is called \emph{weakly optimal} if there exists a separable state $\sigma$ satisfying
\begin{equation}\label{eq:woc}
\Tr(\sigma W)=0.
\end{equation}
We note that weak optimality is a necessary but not sufficient condition for an EW to be optimal\cite{badz2013,chru2014}.

\begin{proof}[Proof of Theorem~\ref{thm:sts}]
As we already have $\pazocal{S}_{j}(\mathscr{E})\leqslant\pazocal{Q}_{j}(\mathscr{E})$, it is enough to show that
\begin{equation}\label{eq:sjts}
\pazocal{Q}_{j}(\mathscr{E})\leqslant\pazocal{S}_{j}(\mathscr{E}).
\end{equation}
Here, we use $\supp(E)$ to denote the support of $E\in\mathbb{H}$\cite{supp}.

Let us consider the set,
\begin{equation}\label{eq:osc}
\Omega_{j}(\mathscr{E})=\big\{\big(1-\Tr E,\,\Tr[\sqrt{\rho_{0}}^{-1}\eta_{j}\rho_{j}\sqrt{\rho_{0}}^{-1}E]-p,\,
E-\sqrt{\rho_{0}}M\sqrt{\rho_{0}}\big)\in\mathbb{R}^{2}\times\mathbb{H}^{\circ}\,\big|\,
p>\pazocal{S}_{j}(\mathscr{E}),~E\in\mathbb{H}_{+}^{\circ},~M\in\mathbb{SEP}\big\}
\end{equation}
where 
\begin{equation}\label{eq:hhpd}
\mathbb{H}^{\circ}=\{E\in\mathbb{H}\,|\,\supp(E)\subseteq\supp(\rho_{0})\},~~
\mathbb{H}_{+}^{\circ}=\{E\in\mathbb{H}_{+}\,|\,\supp(E)\subseteq\supp(\rho_{0})\}.\end{equation}
Due to the convexity of $\mathbb{H}_{+}^{\circ}$ and $\mathbb{SEP}$, $\Omega_{j}(\mathscr{E})$ is a convex set.
Moreover, $\Omega_{j}(\mathscr{E})$ does not have the origin $(0,0,\mathbb{O})$ of $\mathbb{R}^{2}\times\mathbb{H}^{\circ}$, otherwise there exists $M\in\mathbb{SEP}$ satisfying
\begin{equation}\label{eq:rmjs}
\Tr(\rho_{0} M)=1,~~\eta_{j}\Tr(\rho_{j}M)>\pazocal{S}_{j}(\mathscr{E}),
\end{equation}
which contradict Theorem~\ref{thm:lsm}.

The Cartesian product $\mathbb{R}^{2}\times\mathbb{H}^{\circ}$ can be considered 
as a real vector space with an inner product defined as
\begin{equation}\label{eq:inrs}
\langle (a,a',A),(b,b',B)\rangle=ab+a'b'+\Tr(AB)
\end{equation}
for $(a,a',A),(b,b',B)\in\mathbb{R}^{2}\times\mathbb{H}^{\circ}$.
Since the set $\Omega_{j}(\mathscr{E})$ in Eq.~\eqref{eq:osc} and the single-element set $\{(0,0,\mathbb{O})\}$ are disjoint convex sets, it follows from the \emph{separating hyperplane theorem}\cite{boyd2004,sht} that there exists $(\gamma_{1},\gamma_{2},\Gamma)\in\mathbb{R}^{2}\times\mathbb{H}^{\circ}$ satisfying
\begin{eqnarray}
&&(\gamma_{1},\gamma_{2},\Gamma)\neq (0,0,\mathbb{O}),\label{eq:oino}\\
&&\langle(\gamma_{1},\gamma_{2},\Gamma),(r,r',R)\rangle\leqslant 0\label{eq:shyr}
\end{eqnarray}
for all $(r,r',R)\in\Omega_{j}(\mathscr{E})$.

Suppose
\begin{eqnarray}
&&\gamma_{1}\leqslant\gamma_{2}\pazocal{S}_{j}(\mathscr{E}),\label{eq:gg12}\\
&&\gamma_{1}\rho_{0} -\gamma_{2}\eta_{j}\rho_{j}\in\mathbb{SEP}^{*},\label{eq:sbgo}\\
&&\gamma_{2}>0.\label{eq:pgc2}
\end{eqnarray}
From Conditions~\eqref{eq:sbgo} and \eqref{eq:pgc2}, the real number $q=\gamma_{1}/\gamma_{2}$ satisfies
\begin{equation}
q\rho_{0}-\eta_{j}\rho_{j}\in\mathbb{SEP}^{*}.
\end{equation}
Thus, Inequality~\eqref{eq:sjts} holds because
\begin{equation}
\pazocal{Q}_{j}(\mathscr{E})\leqslant 
q=\frac{\gamma_{1}}{\gamma_{2}}\leqslant \pazocal{S}_{j}(\mathscr{E}),
\end{equation}
where the first inequality is from the definition of $\pazocal{Q}_{j}(\mathscr{E})$ in Eq.~\eqref{eq:tsie} and the second inequality is due Conditions~\eqref{eq:gg12} and \eqref{eq:pgc2}.
The rest of this proof is to prove Conditions~\eqref{eq:gg12}, \eqref{eq:sbgo}, and \eqref{eq:pgc2}.
\end{proof}
\begin{proof}[Proof of \eqref{eq:gg12}]
Inequality~\eqref{eq:shyr} can be rewritten as
\begin{equation}\label{eq:ggtt}
\gamma_{1}-\gamma_{2} p \leqslant 
\Tr\big[(\gamma_{1}\mathbbm{1}-\gamma_{2}\sqrt{\rho_{0}}^{-1}\eta_{j}\rho_{j}\sqrt{\rho_{0}}^{-1}-\Gamma)E\big]
+\Tr(\sqrt{\rho_{0}}\Gamma\sqrt{\rho_{0}}M)
\end{equation}
for all $p>\pazocal{S}_{j}(\mathscr{E})$, all $E\in\mathbb{H}_{+}^{\circ}$ and all $M\in\mathbb{SEP}$.
If $E=M=\mathbb{O}$, Inequality~\eqref{eq:ggtt} becomes Inequality~\eqref{eq:gg12} by taking the limit of $p$ to $\pazocal{S}_{j}(\mathscr{E})$.
\end{proof}
\begin{proof}[Proof of \eqref{eq:sbgo}]
For any $M\in\mathbb{SEP}$, Inequality~\eqref{eq:ggtt} becomes 
\begin{equation}\label{eq:gmgi}
\gamma_{1}-\gamma_{2}\pazocal{S}_{j}(\mathscr{E})\leqslant
\Tr[(\gamma_{1}\rho_{0}-\gamma_{2}\eta_{j}\rho_{j})M]
\end{equation}
by fixing $E=\sqrt{\rho_{0}}M\sqrt{\rho_{0}}$ and taking the limit of $p$ to $\pazocal{S}_{j}(\mathscr{E})$.

Assume that
\begin{equation}\label{eq:grmg}
\gamma_{1}\rho_{0}-\gamma_{2}\eta_{j}\rho_{j}\notin\mathbb{SEP}^{*}.
\end{equation}
This assumption implies that there exists $M'\in\mathbb{SEP}$ satisfying
\begin{equation}\label{eq:etss}
\Tr[(\gamma_{1}\rho_{0}-\gamma_{2}\eta_{j}\rho_{j})M']<0.
\end{equation}

Thus, $M=tM'$ for $t>0$ satisfies Inequality~\eqref{eq:gmgi}, that is,
\begin{equation}\label{eq:gm12}
\gamma_{1}-\gamma_{2}\pazocal{S}_{j}(\mathscr{E})\leqslant
t \Tr\Big[(\gamma_{1}\rho_{0}-\gamma_{2}\eta_{j}\rho_{j})M'\Big].
\end{equation}
Since the above inequality is true for arbitrary large $t>0$, $\gamma_{1}-\gamma_{2}\pazocal{S}_{j}(\mathscr{E})$ can also be arbitrary small.
However, this contradicts that $\gamma_{1}$, $\gamma_{2}$ and $\pazocal{S}_{j}(\mathscr{E})$ are bounded.
Thus, Inclusion~\eqref{eq:sbgo} is satisfied.
\end{proof}
\begin{proof}[Proof of \eqref{eq:pgc2}]
To show $\gamma_{2}\geqslant0$, we assume $\gamma_{2}<0$.
If $E=M=\mathbb{O}$, Inequality~\eqref{eq:ggtt} becomes
\begin{equation}\label{eq:gli1}
\gamma_{1}\leqslant -\infty
\end{equation}
by taking the limit of $p$ to $\infty$.
This contradicts that $\gamma_{1}$ is bounded.
Thus, we have
\begin{equation}\label{eq:ggsz}
\gamma_{2}\geqslant0.
\end{equation}

Now, let us assume $\gamma_{2}=0$.
In this case, Inequality~\eqref{eq:gg12} and Inclusion~\eqref{eq:sbgo} become
\begin{equation}\label{eq:grim}
\gamma_{1}\leqslant 0,~\gamma_{1}\rho_{0}\in\mathbb{SEP}^{*}.
\end{equation}
From Eq.~\eqref{eq:grim} together with Inequality~\eqref{eq:twz}, we have
\begin{equation}\label{eq:gaz1}
\gamma_{1}=\Tr(\gamma_{1}\rho_{0})=0.
\end{equation}

For any $E\in\mathbb{H}_{+}^{\circ}$, it follows from $\gamma_{1}=\gamma_{2}=0$ that Inequality~\eqref{eq:ggtt} becomes
\begin{equation}\label{eq:mgeg}
\Tr(-\Gamma E)\geqslant0
\end{equation}
by fixing $M=\mathbb{O}$.  
From Inequality~\eqref{eq:mgeg} and $\Gamma\in\mathbb{H}^{\circ}$, we have
\begin{equation}\label{eq:mghc}
-\Gamma\in\mathbb{H}_{+}^{\circ}
\end{equation}
which implies
\begin{equation}\label{eq:trgp}
\Tr(\sqrt{\rho_{0}}\Gamma\sqrt{\rho_{0}})\leqslant0.
\end{equation}

For any $M\in\mathbb{SEP}$, Inequality~\eqref{eq:ggtt} becomes
\begin{equation}\label{eq:srgm}
\Tr\big(\sqrt{\rho_{0}}\Gamma\sqrt{\rho_{0}}M\big)\geqslant 0
\end{equation}
by fixing $E=\mathbb{O}$.  
This inequality implies
\begin{equation}\label{eq:srgr}
\sqrt{\rho_{0}}\Gamma\sqrt{\rho_{0}}\in\mathbb{SEP}^{*}.
\end{equation}
From Inequality~\eqref{eq:trgp} and Inclusion~\eqref{eq:srgr} together with Inequality~\eqref{eq:twz}, we have
\begin{equation}\label{eq:rgrm}
\sqrt{\rho_{0}}\Gamma\sqrt{\rho_{0}}=\mathbb{O}
\end{equation}
which means 
\begin{equation}\label{eq:gmbo}
\Gamma=\mathbb{O}. 
\end{equation}
Since $(\gamma_{1},\gamma_{2},\Gamma)=(0,0,\mathbb{O})$ contradicts Condition~\eqref{eq:oino},
Inequality~\eqref{eq:pgc2} holds.
\end{proof}

\section*{Acknowledgments}
This work was supported by a National Research Foundation of Korea(NRF) grant funded by the Korean government(Ministry of Science and ICT)(No.NRF2023R1A2C1007039).



\begin{thebibliography}{99}
\bibitem{supp}  
The support of a Hermitian operator is the subspace spanned by the all eigenvectors with nonzero eigenvalue.
\bibitem{pptp}  
The eigenvalues of a partially transposed operator do not depend on the choice of basis or the subsystem to be transposed. 
For simplicity, we consider the standard basis and the second system throughout this paper.
\bibitem{sht}  
When $A$ and $B$ are disjoint convex sets in a real vector space $V$ with an inner product $\langle\cdot,\cdot\rangle$,
there exist $x\in\mathbb{R}$ and $\vec{v}\in V$ such that $\vec{v}\neq\vec{0}$ and $\langle \vec{a},\vec{v}\rangle\leqslant x\leqslant\langle \vec{b},\vec{v}\rangle$  for all $\vec{a}\in A$ and all $\vec{b}\in B$.
\end{thebibliography}
\end{document}